\documentclass[USenglish]{lipics-v2016}
\usepackage [left,pagewise] {lineno}
\usepackage{tikz}
\usetikzlibrary{arrows.meta,decorations,decorations.pathreplacing,calc,bending,positioning, chains,patterns}
\author[1]{Frank Kammer}
\author[1]{Andrej Sajenko}
\definecolor {infocolor} {rgb} {0.6,0.6,0.6}

\theoremstyle{fact}
\newtheorem{fact}{Fact}

\usepackage{color}
\usepackage{xcolor}

\usepackage[inline]{enumitem}

\long\def\full#1{%
 { {#1}}
}

\long\def\conf#1{%
}

\EventEditors{}
\EventLongTitle{}
\EventShortTitle{}
\EventAcronym{}
\EventYear{}
\EventDate{}
\EventLocation{}
\EventLogo{}
\SeriesVolume{}
\ArticleNo{\ }

\newcommand{\Nat}{I\!\!N}
\definecolor{darkgreen}{rgb}{0.0,0.5,0.0}
\definecolor{darkblue}{rgb}{0.0,0.0,0.7}
\definecolor{lightblue}{rgb}{0.2,0.2,1.0}
\definecolor{notgreenish}{rgb}{0.7,0.5,0.0}
\definecolor{greenish}{rgb}{0.2,0.5,0.2}

\title{Extra Space during Initialization of Succinct Data
Structures and Dynamical Initializable Arrays}

\titlerunning{Extra Space during Initialization of Succinct DS and 
Dynamical Initializable Arrays} 

\affil[1]{THM, University of Applied Sciences Mittelhessen, 
Germany\\
  \texttt{\{frank.kammer,andrej.sajenko\}@mni.thm.de}}
\authorrunning{F.\ Kammer and A.\ Sajenko} 
\Copyright{Frank Kammer and Andrej Sajenko}


\begin{document}

\maketitle{}%

  \begin{abstract}%
Many succinct data structures on the word RAM require precomputed tables to
start operating. 
Usually, the tables can be constructed in sublinear time.
In this time, most
of a data structure is not initialized, i.e., there is plenty of
unused space allocated for the data structure.
We present a general framework 
 to store 
temporarily extra buffers between the real data 
so that the data can be processed immediately, stored first in the buffers, 
and then moved into the real data structure after finishing the tables.
As an application, we apply our framework to 
Dodis, P{\v a}tra{\c s}cu, and
               Thorup's
data structure (STOC 2010) that emulates $c$-ary memory and to 
Farzan and Munro's 
succinct encoding of arbitrary graphs (TCS 2013).
We also use our framework to 
present an
in-place dynamical initializable array.
\end{abstract}

\section{Introduction}
Small mobile devices, embedded systems, and big data 
draw the attention to space- and 
time-efficient 
algorithms,
e.g., for sorting \cite{Bea91, 
PagR98}, geometry \cite{AsaBBKMRS14, BarKLSS13,ElmK16} 
or graph algorithms
\full{\cite{AsaIKKOOSTU14,ChaMRS17, DatKM16, ElmHK15, HagK16, HagKL18,KamKL16,KamS18a}.}%
\conf{\cite{AsaIKKOOSTU14,DatKM16, ElmHK15, HagKL18}.} 
Moreover, there has been also an increased
interest in
succinct (encoding) data
structures~\cite{DodPT10,FarM13,HagK16,MunRRR12, MakN08, Pat08, Pag01}.

On a word RAM, succinct 
data structures 
often 
require precomputed tables to
start operating, e.g., 
Dodis, P{\v a}tra{\c s}cu, and
               Thorup's
data structure~\cite{DodPT10}. It 
 emulates
$c$-ary memory, for an arbitrary $c\ge 2$,
on standard binary memory almost without
losing space. Before we can store any information in the data structure,
suitable lookup tables have to be computed. Hagerup and
Kammer~\cite[Theorem 6.5]{HagK16} showed
a 
solution that
allows us to store the incoming information in an extra buffer and
 read and write values immediately.
However, their solution scrambles the data so that extra mapping tables
have to be built and used forever to access the data, even after the lookup tables are built.
Moreover, Farzan and Munro~\cite{FarM13} described a succinct
encoding of arbitrary graphs. To store a dense graph, an adjacency matrix
is stored in a compact form by decomposing the matrix in tiny submatrices.
Each submatrix is represented by an index and the mapping is done via a lookup table. 
Assume that the tables 
are not ready, but the information of the graph
already arrives in a stream. One then can use an extra buffer and
store the small matrices
in a non-compact way and operate on these matrices until the tables are
ready. At the end, the data can be moved from the buffers to the real data
structure.

Usually, the tables can be constructed in sublinear time.
In this time, most
of the data structure is not initialized, leaving plenty of
allocated space unused.
We present a general framework to store a
temporarily extra buffer between the real data during a short startup time
and the buffer has the standard access operations of an array.
In particular, 
if the buffer is not needed any more and thus is removed,
our implementation ''leaves'' the data structure as 
if the extra buffer has never existed.

As another application of our framework, we show that 
in-place initializable arrays can be made dynamic and with every increase
of the array size, the new space is always 
initialized.
Our model of computation 
is the word RAM model with a word length $w$ 
that allows us to access all input words in $O(1)$ 
time. Thus, to operate on a dynamic initializable array of maximum size 
$n_{\mathrm{max}}$ we require that the word size $w = \Omega( \log n_{\mathrm{max}})$.

To obtain a dynamic initializable array, 
we additionally assume that we can expand and shrink an already allocated memory. 
The memory can come from an operating system or can be user controlled,
i.e., the user can shrink an initializable array temporarily, use the space for  
some other purpose, and can increase it
again if it is needed later.

\subsection{Previous Array Implementation}
The folklore algorithm~\cite{Ben86,Meh84}, which uses 
two arrays with pointers pointing to each other and one array storing
the data, uses $O(nw)$ bits of extra space.
Navarro~\cite{Nav14} showed an implementation that requires $n + o(n)$ bits
of extra space. Recently, Hagerup and Kammer \cite{HagK17} 
showed that $\lceil(n(t/(2 
w))^t)\rceil$ extra bits for an arbitrary constant $t$ suffice 
if one wants to support access to the array in $O(t)$ time. In particular,
by choosing $t=O(\log n)$, the extra space is only one bit.
Very recently, Katoh and Goto~\cite{KatG17} showed that 
also constant access time is possible with 
one extra bit. All these array implementations are designed 
for static array 
sizes.

\subsection{Our Contributions}
First, we present our framework to 
extend a data structure by 
an extra buffer.
%
We then apply the framework to two known algorithms~\cite{DodPT10,FarM13}.
Finally, we 
make Katoh and Goto's in-place initializable arrays dynamic.
For this purpose, we identify problems that happen 
by a simple increase of the memory used for the array and 
stepwise improve the 
implementation to make the array 
dynamic. 
Our best solution supports 
operations \textsc{read}, \textsc{write}, and \textsc{increase} in 
constant time and an operation \textsc{shrink} in amortized constant time.

Why do we support the shrink operation only in amortized time?  
The rest of this paragraph is not a precise proof, 
but gives some intuition what the problems are 
if the goal is to support 
constant non-amortized time for initialization, writing and shrinking.
To support constant-time initialization of an array
we must have knowledge of the regions 
that are completely written by the user and these regions must store the
information which words are written. 
If we now allow arbitrary shrink operations, then we do not have the space 
to keep all information. Instead, we must clean the information, but keep
this information of the written words that still belong to the dynamic array.
If we want to run the shrink operation in non-amortized constant time, the
information must be sorted in such a way that the cleaning can be done by
simply cutting away a last part of the information. 
This means the writing operation must take care that the information is stored in some sense
``sorted''. Since
we can not sort a stream of elements 
in $O(1)$ time per element in general, 
it seems plausible that 
the writing operation takes $\omega(1)$ time.

The remainder of this paper is structured as follows.
We begin with summarizing the important parts of Katoh and Goto's algorithm 
in Section~\ref{sec:in-place-array}.
In Section~\ref{sec:extra-datastructures} we present 
our framework
to implement the extra buffer 
stored 
inside the unused space of a data structure 
without using extra space.
In Section~\ref{sec:extraSpdynArr}, we extend our framework to 
dynamic arrays.
As an application, 
we show in
Section~\ref{sec:unintended-chains} how to increase the array size 
in constant 
time.
%
%
In Sections~\ref{sec:complex-solution} and~\ref{sec:general-case}, we present the final implementation 
that also allows us to decrease the array size.
\conf{Some proofs in Section~\ref{sec:dyn-arrays} are omitted, but can be found in the appendix and
in a full version of the paper~\cite{KamS18b}.}

\section{In-Place Initializable Array}\label{sec:in-place-array}
Katoh and Goto~\cite{KatG17} introduced 
the data structure below 
and gave an implementation that uses $nw + 1$ bits and supports all operations 
in constant time. 

\begin{definition}\label{def:init-array}
	An initializable array $\mathcal D$ is a data structure that 
        stores $n \in \Nat$ elements of the universe $Z \subseteq \Nat$
	and provides the following operations: 
	
	\begin{itemize}
		\item \textsc{read}($i$) ($i \in \Nat$): Returns the $i$-th element of $\mathcal D$.
		\item \textsc{write}($i, x$) ($i \in \Nat, x \in Z$): Sets the $i$-th element of $\mathcal D$ to 
		$x$.
		\item \textsc{init}($x$) ($x \in Z$): Sets all elements of $\mathcal D$ to $\mathrm{initv} :=x$.
	\end{itemize}
\end{definition}

Let
$D[0,\ldots,n]$ be an array storing $\mathcal D$. 
$D[0]$ is used only to store one bit 
to check if the array $\mathcal D$ is fully initialized. If $D[0] = 1$,
$\mathcal D$
is a normal array, otherwise the following rules apply.

The idea is to split a standard array into blocks of $b = 2$ words
and group the blocks into two areas: The blocks before some threshold
 $t\in \Nat$ are in the {\em written area}, the remaining in the {\em unwritten
 area}. 
Moreover, they call two blocks {\em chained} 
if the values of their first words point at each others
position and if they belong to different areas. In the following, we denote by
$d(B)$
the block chained with a block $B$. Note that $d(d(B)) = B$.

In our paper, we also use chains. Therefore, we shortly describe the meaning of a
(un)chained block $B$ in~\cite{KatG17}. 
Compare the following description with Figure~\ref{fig:blockStates}.
If $B$ is in the written area and chained,
then $B$ is used to store a value of another initialized block and is therefore defined as 
uninitialized. 
If
$B$ is in the unwritten area and chained, 
$B$ contains
user written values and the words of $B$ are divided among $B$ and $d(B)$.
In this case $d(B)$ was not written by the user. 
In detail, the second word of $B$ is stored in this word of $B$,
 but the first word is stored in the
second word of $d(B)$ since the first word is used to store a pointer.  
If $B$ is in the written area and unchained,
$B$ contains user written (or initial) values. 
If $B$ is in the unwritten area and unchained, 
$B$ has never been
written by the user.
Every time the user writes into a block for the first time, i.e., 
into a word of an (un)chained block in the (un)written area, 
the written area increases by moving the first block of the unwritten area into the
written area and by possibly building or correcting chains.
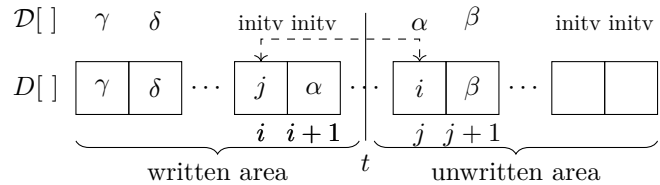
\begin{figure}[h]
	\centering
	\begin{tikzpicture}[
	start chain = going right,
	node distance = 0pt,
	box/.style={draw, minimum width=2em, minimum height=2em,outer sep=0pt, on chain},
	dots/.style={minimum width=2em, minimum height=2em,outer sep=0pt, on chain}],
	\node [box] (1) {$\gamma$};
	\node [box] (2) {$\delta$};
	\node[above=3mm of 1] (g) {$\gamma$};
	\node[above=3mm of 2] (d) {$\delta$};
	\node [dots] (3) {$\cdots$};
	\node [box] (4) {$j$};
	\node [box] (5) {$\alpha$};
	\node[above=3mm of 4] (m)  {\scalebox{0.85}{initv}};
	\node[above=3mm of 5] (n) {\scalebox{0.85}{initv}};
	\node[below=0mm of 4] (i) {$i$};
	\node[below=0mm of 5] (x) {$i + 1$};
	\node [dots] (6) {$\cdots$};
	\draw [draw, on chain] (3.5,-0.7) -- (3.5,1);
	\node (t) at (3.5,-1) {$t$};
	\node [box] (7) {$i$};
	\node [box] (8) {$\beta$};
	\node[above=3mm of 7] (a) {$\alpha$};
	\node[above=3mm of 8] (b) {$\beta$};
	\node[below=0mm of 7] (j) {$j$};
	\node[below=0mm of 8] (y) {$j + 1$};
	\node [dots] (9) {$\cdots$};
	\node [box] (10) {};
	\node [box] (11) {};
	\node[above=3mm of 10] (t) {\scalebox{0.85}{initv}};
	\node[above=3mm of 11] (s) {\scalebox{0.85}{initv}};
	\node[below=0mm of 4] (i) {$i$};
	\node[below=0mm of 5] (i) {$i + 1$};
	\begin{scope}[-{Stealth[length = 2.5pt]}]
	\node[above=2mm of 4] (e) {};
	\node[above=2mm of 7] (r) {};
	\path [<->,dashed, draw](4.north) -- ++(0.0,0.3 ) -|  (7.north);
	\node[left=1mm of 1] (d) {$D$[ ]};
	\node[above=3mm of d] (D) {$\mathcal D$[ ]};
	\end{scope}
	\draw[decorate,decoration={brace, amplitude=5pt, raise=10pt, mirror}]
	(1.south west) to node[black,midway,below= 15pt] {written area} ($(6.south) + (-1mm,0)$);%
	\draw[decorate,decoration={brace, amplitude=5pt, raise=10pt, mirror}]
	($(6.south) + (1mm,0)$) to node[black,midway,below= 15pt] {unwritten area} (11.south east);%
	\end{tikzpicture}
	\caption{The different states of blocks inside the written and unwritten area.
		$\mathcal D$ represents the users view on the data with {\em initv}
 being the initial value, $\alpha, 
		\beta, \gamma$ and $\delta$ as user written values. $D$~represents the internal view with $i$ 
		and $j$ as indices of $D$.}\label{fig:blockStates}
\end{figure}

\newpage
\section{Extra Space during Initialization of Succinct Data Structures}\label{sec:extra-datastructures}

We now consider an arbitrary data structure $\mathcal D^*$. 
The only assumption that we make is that
$\mathcal D^*$ accesses the memory via a data structure $\mathcal D$ realizing
an $n$-word array.
For the time being, assume $\mathcal D$ is the data structure 
as described in Section~\ref{sec:in-place-array} with $D[0,\ldots,n]$ being an
 array storing $\mathcal D$.

We define $|U|$ as the number of blocks of the unwritten area.
As long as $\mathcal D$ is not fully initialized, and only $m \in 
\Nat$ with $m < n$ words are written in $\mathcal D$, Theorem~\ref{lem:place-in-array} shows that 
the unused storage allocated for 
$\mathcal D$ allows us to store information of an extra array $E$ inside $D$ by increasing 
the block size to $b > 2$ (Figure~\ref{fig:EInD}).
The size of $E$ depends on $|U|$ and thus on the current grade of 
initialization of $\mathcal D$.
In contrast to $\mathcal D$, $E$ is an uninitialized array.
Clearly, we can build an initialized array on top of $E$.

The array size $n$ is not always a multiple of the block size $b$.
By decreasing $n$ and handling less than $b$ words of $\mathcal D$
separately, we assume in this paper that $n$ is a multiple of $b$.

\begin{theorem}\label{lem:place-in-array}
	We can store an extra array $E$ of $(n/b - m)(b - 2) \le |U|(b - 2)$ words inside 
	the unused space of~$\mathcal D$. $E$ dynamically shrinks from the
 end whenever $m$ increases.
\end{theorem}
\begin{proof} 
        Let us consider a block $B$ inside the unwritten area.
        $B$ and $d(B)$ together provide 
	space for $2b$ words, but the chain between them exists only because the user wrote inside
	one block (block $B$, but not in $d(B)$), and thus we need to store $b$ words
        of user data. 
        Furthermore, $B$ and $d(B)$ each needs one word to store a pointer that represents 
	the chain. 
        By storing as much user data as possible of $B$ in $d(B)$, we have $b - 2$
	unused 	words in $B$---say, words $2$ to $b-1$ are unused. 
        Note that, if a block $B$ is unchained, then it contains no user values
	at all and we have even $b$ unused words.
	If the user wrote $m$ different words in $\mathcal D$, then the threshold $t$ of $\mathcal D$ ($t$ 
	is the number of blocks in the written area) is expanded at most 
	$m$-times.
        The written area consists of at most $m$ blocks of a total of $n / b$
	blocks. Consequently, the unwritten area has $(n/b - m) \le |U|$ blocks with each having $(b - 2)$ unused 
	words.
If we start to store $E[0], E[1], ...$ in the unwritten area strictly behind $t$, 
the indices of $E$ will shift 
every time 
the unwritten area 
shrinks.
Therefore, we store $E$ in reverse and $E$ loses with every shrink the last $b - 2$ words\full{, but get 
static 
indices}.
%
\end{proof}

\begin{figure}[h!]
	\centering
	\begin{tikzpicture}[
	start chain = going right,
	node distance = 0pt,
	fbox/.style={draw, minimum width=2em, minimum height=2em,outer sep=0pt},
	box/.style={draw, minimum width=2em, minimum height=2em,outer sep=0pt, on chain},
	gray/.style={fill=gray!45},
	dots/.style={minimum width=2em, minimum height=2em,outer sep=0pt, on chain}],
	\node [box] (0) {};
	\node [dots] (3) {$\cdots$};
	\node [box] (4) {$j$};
	\node[below=0mm of 4] (j1) {$i$};
	\node [box] (a) {$\alpha$};
	\node [box] (5) {$\beta$};
	\node [dots] (9d) {$\cdots$};
	\draw [draw, on chain] (3.85,-0.7) -- (3.85,1);
	\node (t) at (3.85,-1) {$t$};
	\node [box] (7) {};
	\node [box] (b) {\scalebox{0.8}{$E$[$|U|{-}1$]}};
	\node [box] (8) {};
	\node [dots] (9) {$\cdots$};
	\node [box] (d1) {$j$};
	\node [box] (d) {\scalebox{0.8}{$E$[\ ]}};
	\node [box] (d2) {$\gamma$};
	\node[above=3mm of d1] (a1) {$\alpha$};
	\node[above=3mm of d] (b1) {$\beta$};
	\node[above=3mm of d2] (c1) {$\gamma$};	
	\node[below=0mm of d1] (i1) {$i$};
	\path [<->,dashed, draw](4.north) -- ++(0.0,0.25) -|  (d1.north);
	\node [dots] (13d) {$\cdots$};
	\node [box] (10) {};
	\node [box] (c) {\scalebox{0.8}{$E$[0]}};
	\node [box] (12) {};	
	\begin{scope}[-{Stealth[length = 2.5pt]}]
	\node[left=1mm of 0] (d) {$D$[ ]};
	\node[above=2mm of d] (D) {$\mathcal D$[ ]};
	\end{scope}
	\draw[decorate,decoration={brace, amplitude=5pt, raise=10pt, mirror}]
	(0.south west) to node[black,midway,below= 15pt] {written area} ($(7.south west) + (-1mm,0)$);%
	\draw[decorate,decoration={brace, amplitude=5pt, raise=10pt, mirror}]
	($(7.south west) + (0.5mm,0)$) to node[black,midway,below= 15pt] {unwritten area} (12.south 
	east);%
	\end{tikzpicture}
	\caption{
          The block size $b = 3$ allows us to move $\alpha$ and $\beta$ to the chained
          block such that each block of the unwritten area has one unused word 
          to store the extra data of 
          $E$.}\label{fig:EInD}
\end{figure}
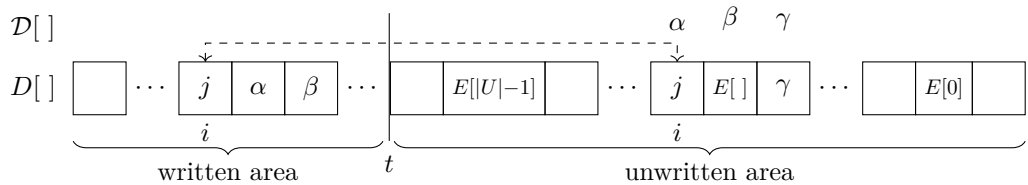

%
Note that values like 
the threshold $t$ and the initial value can be easily stored at the beginning of $E$. 
However, we require the size of $\mathcal D$ to determine the beginning of $E$.
To store the size we therefore move $D[1]$
also into $E$ and store the size of the array in $D[1]$.
During the usage of $D$ we have to take this into account, 
but for simplicity we ignore this fact.
If the array is fully initialized and thus $D[0] = 1$, we can not store the size
anymore, but in this case we do not need to know it.

\subsection{Application: $c$-ary Memory}

To use Dodis et al.'s dictionary~\cite{DodPT10} on $n$ elements, a lookup
table $Y$ with $O(\log n)$ entries consisting of $O(\log n)$ bits each must
be constructed, which can be done in $O(\log n)$ time.  Using tables of the
same kind to store powers of $c$, we can assume that $c=\Omega(n)$ by
combining consecutive elements of the input into one element.

Hagerup and Kammer extended Dodis et al.'s dictionary 
to support constant-time initialization by storing the 
$k=\Theta(\log n)$
elements that are added first to the dictionary in a trie $D^\mathrm{T}$ 
of constant depth $d\ge
4$ and out degree 
$n^{1/d}$ using $O(n^{\epsilon+1/d}+\log n \log c)$ bits for any $\epsilon>0$. Interleaved with
the first $k$ operations on the dictionary, first the table $Y$ is built and afterwards, 
the elements in $D^\mathrm{T}$
are moved to the dictionary.


Since $D^\mathrm{T}$ is required only
temporarily, it is stored within the memory allocated for the 
dictionary---say in the last part of the memory.  
The problem that arises is that 
the
last part of the memory can not be used as long as it is still used by $D^\mathrm{T}$. 
This problem is solved by partitioning the memory allocated for the 
dictionary into sectors and using a complicated mapping function that 
scrambles
the elements to avoid the usage of the last part of the memory. Unfortunately,
even after computing the table $Y$ and moving all elements from
$D^\mathrm{T}$ over to the dictionary, the data is still scrambled and 
each access to the dictionary has to start evaluating the mapping function.

With Theorem~\ref{lem:place-in-array} as an underlying data structure it is
easily possible to implement the dictionary with constant initialization time.
%
Use
the extra array $E$ to store temporarily the table $D^\mathrm{T}$. 
Even if we use block size $b = 3$, 
$E$ has enough space (at least $(n (\log c) - k \log n)/3 \ge (n-k)(\log n)/3$ 
bits at the end of the construction
of $Y$) to store
$D^\mathrm{T}$. 
The mapping
function becomes superfluous\full{ because the data is not scrambled by the usage of 
Theorem~\ref{lem:place-in-array}}.
%
Furthermore, working on a copied and slightly modified algorithm after
the full initialization of $\mathcal D$, 
we can avoid checking the extra bit in $D[0]$.

\subsection{Application: Succinct Encoding of Dense Graphs}
Farzan and Munro~\cite{FarM13} showed a succinct
encoding of an $n \times n$-matrix that represents an arbitrary graph with $n$ vertices and $m$ 
edges. 
Knowing the number of edges they distinguish between five cases: 
An almost full case, where the matrix consists of almost only one entries, an extremely dense case, a dense 
case, a 
moderate case and an extremely sparse case. 
For each case, they present a succinct encoding that supports the query 
operations for adjacency and degree in constant time and iteration over neighbors of a 
vertex in constant~time~per~neighbor.

We consider only the dense case where table lookup is used. Dense means that
$\exists \delta > 0 : $ $n^2 / \log^{1-\delta} n \le m \le n^2 (1 - 1/(\log^{1 -
\delta} n)$.  
As shown in~\cite{FarM13}, a representation in that case requires $\log {n^2 \choose m} + 
O(n^2(\log^{1-\delta} n)) = 
\log {n^2 \choose m} +
o(\log {n^2 \choose m})$ bits of memory. 
To encode the matrix of a graph Farzan and Munro
first divide the matrix into small {\em submatrices} of size $\log^{1 - \delta}
n \times \log^{1 - \delta} n$ for a constant $0 < \delta \le 1$.  
For each row and column of
this smaller submatrices they calculate a {\em summary bit} that is $1$ if the row and column, 
respectively, of the submatrix contains at least one~$1$.
The summary bits of a row and column, respectively, of the whole matrix is used to create a  {\em 
summary vector}.
On top of each summary vector they build a rank-select data structure~\cite{Pat08} that supports 
queries on $0$'s and $1$'s in constant time.
They also build a lookup table to map between possible submatrices and
indices as well as to answer all queries of interest in the submatrix in $O(1)$ time.
In a further step, they replace each submatrix by its index.
%
%
By guaranteeing that the number of possible submatrices is $o(\log n)$, the construction of the 
lookup table
can be done in $O(n)$ time.
To simply guarantee this, we restrict $\delta$ to be larger than $1/2$.
We generalize the result for $\delta > 1/2$ to dynamic
dense graphs by replacing the rank-select data structures with choice
dictionaries~\cite[Theorem~7.6]{HagK16} as follows.
We assume that the graph has initially no edges and that 
there is a stream that consists of edge updates and query operations.
With each edge update, the index of the submatrix with the edge changes.
To realize the index transition we use a translation table that maps 
both, the current index of the
submatrix and the edge update made, to the new index.
We store for each summary vector an {\em edge counter}, i.e., the number of $1$'s that was used 
to calculate each summary bit of it.
Whenever an edge is created between two nodes we set the corresponding bit of the row and column, 
respectively, inside the summary vector to $1$ and increment the edge counter.
If an edge is removed we decrement the edge counter and determine, using a lookup table, if the 
submatrix containing this edge has still $1$'s in the updated row and column, respectively.
If it does not, we set the corresponding summary bit inside the summary vector to $0$.

If the lookup tables for the submatrices are computed, we can easily 
answer queries to the current graph similar to~\cite{FarM13} since we still use the summary vectors.
We can ask for membership to answer adjacency queries. To iterate over 
the 
neighbors of a vertex (i.e., over the $1$'s in the summary vector), we can use the iterator function of 
the 
choice dictionary 
instead of using the select function of the rank-select data structure.
The degree query is answered by returning the edge counter.

In the rest of this subsection, the goal is to extend the data structure such that it supports the queries 
without a delay for the construction of the lookup tables. The idea is to store a submatrix with a $1$ 
entry
in the usual, non-compact way in the extra buffer and to add a pointer 
from the submatrix to a place in the buffer
where it is stored.
To reduce the space used for the
pointers, we group $\log^{\delta} n$ submatrices together to get a {\em group matrix} 
of size $\log^{1 - \delta} n \times \log n$.
Moreover, we use an array $A$ in
which we store a pointer of $O(\log n)$ bits for each group. 
With the first writing operation to one submatrix in a
group, the group matrix is stored in the usual, non-compact way in 
an initialized array, which is stored in 
the extra
buffer (Theorem~\ref{lem:place-in-array}).
In the array $A$, we update the pointer for the group.
In addition, we build a choice dictionary on top of each 
non-empty column and
each non-empty row of the group matrix. 
With each choice dictionary, we also count the number of $1$ entries. 

Using the counts we can update the summary vectors and 
the non-compact submatrix allows us to answer adjacency queries. Both can
be done in constant time without using lookup tables.
To support the neighborhood query, we use choice dictionaries on top of the
group matrices. 
%
In detail, whenever updating an edge, update the summary vector, add or follow the
pointer to the non-compact representation, update it including the choice dictionary 
and the degree counter.

After $O(n)$ operations the 
lookup table is computed and, in the same time, the entries in the non-compact matrices are
moved to the compact submatrices. 


The array $A$ requires $n^2 (\log n) /((\log^{1-\delta} n)(\log n)) = O(n^2/\log^{1-\delta}
n)$ bits\full{, which is negligible}.  Moreover, a group matrix with the choice
dictionaries requires $O(\log^2 n)$ bits to be stored.  Thus, after
$O(n)$ operations, the total extra space usage is $O(n \log^2 n)$ bits, which
easily fits in the extra buffers of size
$\Omega(\log {n^2 \choose m}  - n\log n) = \Omega(n^2/\log^{1-\delta} n -  n\log
n)$ 
%
bits after $O(n)$ 
operations.

\begin{theorem}
A graph with $n$ vertices can be stored with $log_2 \binom{n}{m}+O(n^2/\log^{1-\delta} n)$ bits
supporting edge updates as well as adjacency and degree queries
in constant time and iteration over
neighbors of a
vertex in constant~time~per~neighbor where $m=n^2/\log_2^{1-\delta} n$ and
$\delta>1/2$. A
startup time for constructing tables is not necessary.
\end{theorem}

Note that the space bound is succinct if the graph is dense. Moreover, if 
the whole memory of a graph representation must be allocated in the
beginning, i.e., if we
do not allow a change of the space bound during the edge updates, 
then our
representation of the dynamical graph (with possibly $n^2/\log_2^{1-\delta}
n$ edges) is also succinct.

\section{Extra Space for Dynamic Arrays}\label{sec:extraSpdynArr}
We now extend our framework to dynamic arrays.
%
%
Then the unused space inside $\mathcal D$ may grow or shrink and with it the size of
the array $E$. 

Since we start to index the words belonging to $E$ from the end of $\mathcal D$, 
changing the size of $\mathcal D$ will change the index of the words 
in $E$.  If we start to index $E$ from the beginning of the unwritten area,
all indices change every time the unwritten area shrinks.

To realize an indexed array we introduce a fixed-length array $F$.  To handle the
shrink operation, we store $F$ strictly behind the written area so that there
is nothing to do if we shrink or expand the size of $\mathcal D$.  If the
written area grows, we lose $b - 2$ unused words of the unwritten area and
therefore $b - 2$ used words in $F$.  As long as there are unused words in
the unwritten area behind $F$ we move the $b-2$ words that we would lose to the
first unused words behind $F$.  This will rotate $F$ inside the unwritten area of $D$.  
%
\begin{lemma}\label{lem:static-array}
	We can store a fixed-length array $F$ consisting of $\ell$ words in the unused space of 
	$\mathcal D$ 
	as long as $\ell \leq (n/b -m - 1)(b - 2) \le (|U| - 1)(b-2)$. If the condition is violated, then $F$ is 
	destroyed.
\end{lemma}
\begin{proof}
	Let the size of $F$ be smaller than the unused space of $\mathcal
	D$.  Whenever the user writes an additional word in $\mathcal D$,
	$m$ increases and the number of the unused words becomes less. 
	Therefore, before extending the written area by one block, we move
	the $b - 2$ words in that block to the unused space of $\mathcal D$
	behind $F$.  We use the last block in the unwritten area to store the
	size and a counter to calculate the rotation.  Thus, we have one
	block less in contrast to Theorem~\ref{lem:place-in-array}.  $F$ can be of size $\ell
	\le (n/b -m - 1)(b - 2) \le (|U| - 1)(b-2)$.
\end{proof}
For some applications it is interesting to have a dynamic extra storage even if $\mathcal D$ is dynamic.
We introduce a dynamic set in Lemma~\ref{lem:dyn-list} that can be stored inside the 
unwritten area of $\mathcal D$ and supports the operations 
\textsc{add}, \textsc{remove}, and 
\textsc{iterate}. The last operation returns a list of 
all elements in the data structure.
The size of $L$ is dynamically upper bounded by $|U|$.
\begin{lemma}\label{lem:dyn-list}
	We can store a dynamic (multi-)set of maximal $\ell \leq (|U| - 1)(b - 2)$ 
        elements
        in the unused space of 
	$\mathcal D$. If the condition is violated by 
	write operations of $\mathcal D$, elements of the set are removed
        until the condition holds again.
\end{lemma}
\begin{proof}
	We store a dynamic list strictly behind the written area and shift it cyclically 
        as in the proof of 
	Lemma~\ref{lem:static-array}. The difference here is that 
        we have no fixed order so that we 
        can store new elements simply at the beginning of the 
	unused words of $\mathcal D$ and increase the size $\ell$. The
	removal of an element may create 
	a gap that can be filled by moving an element and decreasing $\ell$.
\end{proof}
For our implementation described in Section~\ref{sec:dyn-arrays} we require a (multi-)set with 
\textit{direct access} to the elements. 
We can use the position in $D$ as an address for direct access,
but this requires that the user 
implements
a \textit{notification function}. The function is called by 
our data structure whenever an 
element in the list changes its position to inform the user of the new position of the element.

\begin{fact}\label{fact:space-in-chained-block}
	Having two chained blocks $B$ and $d(B)$ and a suitable block size $b$,
	we can rearrange the user data in the two blocks such that we can 
	store $O(1)$ extra words (information as, e.g., pointers) in both $B$ and in $d(B)$.
\end{fact}

An algorithm often uses several data
structures in parallel, which all can be stored in the uninitialized space
of $\mathcal D$.

\begin{lemma}\label{lem:several-data-structures}
  For a $b \geq 2$, $\mathcal D$ can have $k\in \Nat$ extra data structures in parallel
 where each is of size
  at most
  $|U|\lfloor (b - 2)/k \rfloor$. 
\end{lemma}
\begin{proof}
	Every block of the unwritten area has $b - 2$ unused words and every unused word can be 
	used for another data structure.
\end{proof}
As we see later, it is useful to have several dynamic sets of elements that are all the same
size.  
The sets can also be (multi-)sets.
\begin{corollary}\label{cor:multiSet}
	For $b>2$, $\mathcal D$ can store a dynamic family $\mathcal{F}$ consisting of 
	dynamic sets where the sets have in total 
        at most $\lfloor (|U| - 3)(b - 2)/(s + 3)-1\rfloor$ elements each 
	of size $s$.
\end{corollary}
\begin{proof}
	We use the unused words in the unwritten area of $\mathcal D$
	and 
	partition it into sections of
	$s + 3$ words.
	Every first word of a section is used for an array $I$ and 
	the last $s + 2$ words for an array $S$.
	$S$ is 
	used to store the elements of the sets where the elements of each
	set are connected in   
	a doubly linked list. 
	$I[i]$ points to an element of the $i$th set. 
	The unused words in the last 3 blocks (possibly, fewer blocks suffice) 
        of the unwritten area 
	are used to store some constants: the size $t$
	of the written area, the initial value of $\mathcal D$, some 
	parameter $p\in \Nat$, the 
	number of sets in $\mathcal{S}$, and 
	the total number $q$ of items over all sets.
	%
	Using $q$ we can simply find the section of a new element.
	
	If the written area of $\mathcal D$ increases, 
	we have to start moving the first section to a new unused section. We store
	the position of the new section in $p$. 
	By knowing $t$ we know how much of the first section has been
	already moved and where to find the information of the first section.
\end{proof}

\section{Dynamical Initializable Array}\label{sec:dyn-arrays}

In the next three subsections, we provide implementations to make the in-place initializable array 
dynamic. Therefore, we extend the initializable array $\mathcal D$ by the following two functions to 
change the current size of the array.
\begin{itemize}
	\item \textsc{increase}($n_{\mathrm{old}}, n_{\mathrm{new}}, \mathrm{initv}$) ($n_{\mathrm{old}}, 
	n_{\mathrm{new}} \in \Nat$): Sets the size of $\mathcal D$ from $n_{\mathrm{old}}$ to 
	$n_{\mathrm{new}}$.
	All the elements behind the $n_{\mathrm{old}}$th element in $\mathcal D$ are initialized with $\mathrm{initv}$,
	the initial value of $\mathcal D$ defined by the last call of \textsc{init}.
	\item \textsc{shrink}($n_{\mathrm{old}}, n_{\mathrm{new}}$) ($n_{\mathrm{old}}, n_{\mathrm{new}} 
	\in \Nat$): Sets the size of $\mathcal D$ from $n_{\mathrm{old}}$ to $n_{\mathrm{new}}$.
\end{itemize}

In our implementation we handle the array $\mathcal D$ differently, according to its size.
Recall that $w$ is the word size.
As long as $\mathcal D$ consists of $n'=O(w)$ elements, we use
a bit vector stored in
$O(1)$ words to know which words are initialized.
On the word RAM we can manipulate it in constant time.
To have immediate access to the bit vector we store it 
in the beginning of the array and move the data located there 
to the first unused words.
If we increase the size of $\mathcal D$ 
to more than $\omega(w)$ elements, we keep the bit vector in some
unused words until the first $n'$ words of
$\mathcal D$ are completely initialized. 
Having such a bit vector we assume 
that this is taken into account whenever there is a reading
or writing operation, but do not mention it explicitly in the rest of the paper.

If the array $\mathcal D$ consists of at least $\omega(w)$
elements, we can not use the solution with the bit vector. Instead, 
we use chains 
so that we can distinguish between initialized and
uninitialized blocks, even after several shrink and increase operations. 
%
%
We assume in the following that $\mathcal D$ consists of 
$\omega(w)$ elements.

\subsection{Increasing the size of $\mathcal D$}\label{sec:unintended-chains}
An increase of the size of~$\mathcal D$ means allocating additional memory that gives us several new blocks inside the unwritten area. 
These blocks may contain arbitrary values and these values can point at each other such that they 
create a so-called unintended chain between a block $B_w$ 
(written area) and a block $B_u$ (unwritten area). %
%
%
Katoh and Goto {\em break such unintended chains} by creating a self-chain (pointer at its own 
position) with $B_u$ whenever they write a value to a block of the written area.
Because the increase of $\mathcal D$ may be arbitrary, we can not destroy such chains in 
constant time.

To support large increases of $\mathcal D$, we eliminate the possibility of unintended chains by 
introducing 
another 
kind of chain, called \textit{verification chain}.
To distinguish the two kinds of chains we name the chain introduced in 
Section~\ref{sec:in-place-array} a \textit{data chain}. We define a chain between two 
blocks $B_w$ and $B_u$ as {\em intended} exactly if $B_w$ has also a verification
chain\full{, otherwise as
{\em unintended}}. 

To use the verification chain, we first distribute  
the user data among two chained 
blocks such that each block 
of the written area has an unused word 
(Fact~\ref{fact:space-in-chained-block}).
Let $L$ be a dynamic list with direct access stored in~$\mathcal D$ (Lemma~\ref{lem:dyn-list}).
Whenever the block $B$ in the written area has a data chain with a block $d(B)$, we additionally 
require that it also must 
have a verification chain with an element of the list $L$ (Figure~\ref{fig:dynSet}).
We denote by $v(B)$ a pointer to the element in $L$ chained with $B$ 
and use an unused word of $B$ to store $v(B)$.
Recall that, whenever the unwritten area shrinks, some elements in $L$ change their position in $D$. 
To ensure the validity of the verification chains, we use the notification function of $L$ to update the 
pointer $k \in \Nat$ in the effected blocks.
The verification pointers in $L$, stored behind the written area, are not affected by increasing the size of 
$\mathcal D$. 
%

Finally, note that $L$ has enough space to store all verification pointers 
since only one verification chain is required for 
each block in the unwritten area, i.e, $L$ is of size $O(|U|)$.

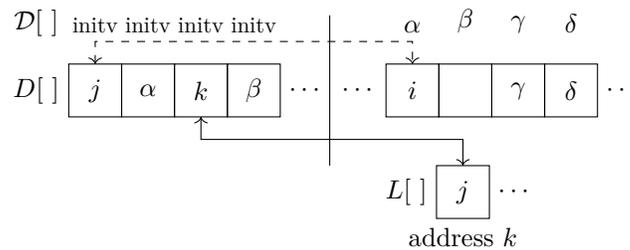
\begin{figure}[h]
	\centering
	\begin{tikzpicture}[
	start chain = going right,
	node distance = 0pt,
	box/.style={draw, minimum width=2em, minimum height=2em,outer sep=0pt, on chain},
	dots/.style={minimum width=2em, minimum height=2em,outer sep=0pt, on chain}],
	\begin{scope}[local bounding box=scope1]
	\node [on chain] (start) {$D$[ ]};
		\node[above=3mm of start] (D) {$\mathcal D$[ ]};
	\node [box] (1a) {$j$};
	\node [box] (1b) {$\alpha$};
	\node [box] (1c) {$k$};
	\node [box] (1d) {$\beta$};
	\node[above=3mm of 1a] (a) {\scalebox{0.85}{initv}};
	\node[above=3mm of 1b] (b) {\scalebox{0.85}{initv}};
	\node[above=3mm of 1c] (b) {\scalebox{0.85}{initv}};
	\node[above=3mm of 1d] (b) {\scalebox{0.85}{initv}};
	\node [dots] (6) {$\cdots$};
	\draw [draw, on chain] (3.9,-1) -- (3.9,1);
	\node [dots] (61) {$\cdots$};
	\node [box] (2a) {$i$};
	\node [box] (2b) {};
	\node [box] (2c) {$\gamma$};
	\node [box] (2d) {$\delta$};
	\node[above=3mm of 2a] (a1) {$\alpha$};
	\node[above=3mm of 2b] (b1) {$\beta$};
	\node[above=3mm of 2c] (c1) {$\gamma$};
	\node[above=3mm of 2d] (d1) {$\delta$};
	\node [dots] (9) {$\cdots$};
	\end{scope}
	\begin{scope}[-{Stealth[length = 2.5pt]}]
	\node[above=2mm of 4] (e) {};
	\node[above=2mm of 7] (r) {};
	\path [<->,dashed, draw](1a.north) -- ++(0.0,0.3) -|  (2a.north);
	\end{scope}
	\begin{scope}[start chain=2, shift={($(2a.south)+(-0.8mm,-1cm)$)}]
		\node[on chain=2] (d) {$L$[ ]}; 
		\node[box, on chain=2](l1) {$j$};
			\node[below=0mm of l1] (as) {address $k$};
		\node [dots, on chain=2] (l2) {$\cdots$};
	\end{scope}
	\path [<->, draw](1c.south) -- ++(0.0,-0.3) -|  (l1.north);

	\end{tikzpicture}
	\caption{The block $B_w$ (left) has a verification chain with an element of a dynamic list $L$ to 
	verify that the data chain between block $B_w$ and $B_u$ (right) is intended.}\label{fig:dynSet}
\end{figure}

\subsection{Shrinking the Size of $\mathcal D$ at most $O(w)$ Times}\label{sec:complex-solution}

Shrinking the size of $\mathcal D$ means to free some memory that may be used
to store information.  We distinguish between two cases.  In the first case
we shrink the size of $\mathcal D$ so much that $L$ is destroyed.  In
this case we make a \textit{full initialization} of $\mathcal{D}$ as
follows.

Since we lose the ability to check for
unintended data chains, we first destroy them
by iterating over the verification list $L$, following its pointer to a block $B_w$, checking if
the block $B_u = d(B_w)$ is behind the new size of $\mathcal D$.
If it is, we initialize $B_w$ with the initial value.
Writing the initial value into $B_w$ may created an unintended data chain, which we break.
Now the verification list is superfluous and we can 
iterate over the unwritten area and check for a data chain. 
If there is one, we move the user values from $d(B)$ into $B$ and
initialize $d(B)$, otherwise we initialize $B$.  
After the iteration, we set $D[0] = 1$. 

Since $L$ will be
destroyed by the shrinking operation, the size $|U|$ of the unwritten area
behind the shrinking is 
bounded by $|L|$.
  Thus, by using the potential function $\psi = (c
\,\cdot\,$length of $L$) for some constant $c\ge 2$, the amortized cost of the \textsc{write} operation 
increases by at most $c=O(1)$ 
and we can easily pay
for the full initialization.

In the second case, we reduce the size of $\mathcal D$ such that the
verification list $L$ is still completely present in $\mathcal D$.  In this
case, all blocks of the written area remain and they still have a
verification chain.  However, the data chain can point outside of $\mathcal
D$.  If we subsequently re-increase $\mathcal D$, some blocks may still have
a data chain and also a verification chain.  This violates our definition of
the \textsc{increase} operation. 

To resolve this problem 
%
we invalidate the data chains of all blocks outside $\mathcal D$ as follows.
As long as there is no shrinking operation, 
we store the same {\em version number} $v \in \Nat$ to all data chains.
More exactly, let $B$ be a block in the unwritten area chained with a block
$d(B)$ in the   
written area. Then we store the version
 inside an unused word of $d(B)$
using Fact~\ref{fact:space-in-chained-block}.

Before we execute a shrink operation, we set $n_{v}$ to the current size of
$\mathcal D$ and remember it as the size for the current version $v$. 
With the shrink operation, we increment the version number by one. 
%
Let $v'$ be the version of the chain between $B$ and $d(B)$.
We call the chain {\em valid} if $B$ is before 
the boundary $n^*_{v'}=\min\{n_v|v\ge v'\}$.
Note that $n^*_{v'}$ is the 
minimal boundary 
that 
$\mathcal D$ ever had after introducing version
number $v'$.

We obtain the 
boundaries $n^*_v$ 
from a data structure $M$ that provides an 
operation to 
add the new size of $\mathcal D$ after a shrink operation and another operation \textsc{minbound} to
check if a chain is valid, i.e., if one endpoint of the chain is or was outside of $\mathcal D$.
For later usage, $M$ also supports an operation \textsc{remove}. 
\begin{itemize}
	\item $\textsc{add}(n)$ ($n \in \Nat$): Increments an internal version counter $v$ by one 
	and set the boundary $n^*_v = n$. All boundaries 
	$n^*_i$ $(i 
	\in 
	\{1 \ldots v - 1\})$ larger than $n$ are overwritten by $n$.
	\item \textsc{remove}($j$) ($0 \le j \le v$): 
        Decrements $v$  by $j$
	and removes the boundaries of the largest $j$ versions.
	\item $\textsc{minbound}(v)$ ($0 \le v < w$): Returns $n^*_v$, the minimal 
	boundary for $v$.
\end{itemize}
\begin{lemma}\label{lem:dictionary}
	$M$ can be implemented such that it uses $O(w)$ words, 
        $\textsc{add}$ runs in amortized constant time, and 
	$\textsc{minbound}$ and \textsc{remove} run in constant 
	time as long as there are only $O(w)$ versions.
\end{lemma}

\long\def\lemDictionary{
\begin{proof}
        We have a version counter $v^*$ and a table $T$ where we store the
        initial boundary for each version.
        Moreover, we use a stack $S$ that additionally allows us to access the elements of the 
	stack directly and the data structure $R$ from P{\v a}tra{\c s}cu and 
	Thorup~\cite{PatT14}. All three data structures can be implemented using a fixed size array
        and stored in extra buffers
        (Lemmas~\ref{lem:static-array}~and~\ref{lem:several-data-structures}).
        $R$ consists of all versions $v$ with $n_v=n^*_v$.
        The stack $S$ store the boundaries of the versions in $R$ in ascending 
        order from the bottom to the top.
        For a technical reason, $S$ stores also the version with each boundary.

        To answer the \textsc{minbound} operation for a version $v$, we first determine
        the successor $v'$ of $v$ in $R$ and then return $n^*_v=T[v']$.

	Assume that a new boundary $n_v = n$ is added to $M$ for a next version $v$.
	Set $T[v] = n$. 
	Before adding a new boundary $n_v$ to $S$ we remove all boundaries from $S$
	as long as the top of $S$ is larger than $n_v$. Whenever removing
	such a boundary, we remove the corresponding version from $R$.
        Finally, we add the new boundary to $S$ and to $R$.
	By doing this, we remove all boundaries that are larger than $n_v$. 
	Now all the versions of these boundaries get $n_v$ as their new boundary. 
	
	Finally, it remains to check the running time. 
	We use the potential function $\phi = |S|$ with $|S| \le w$
        being the current size of $S$. 
	Determining the boundary for a specific version requires looking into $R$
	and to check one word in $S$ and $T$. This does not change $\phi$ and runs in $O(1)$ time.
        If the  running time for adding or removing versions is not
	constant, then in all except a last iteration, we remove an element
	from the stack. Thus, the decrease of $\phi$ pays for this.
	%
\end{proof}}

\full{\lemDictionary}

\full{\begin{corollary}
 For every data chain, we can check in constant time if it is valid or not.
\end{corollary}
\begin{proof}
A chain between some blocks $B$ and $d(B)$ is valid exactly if $d(B) \le \textsc{minbound}(v)$ is true, 
where $v$ is the
version of the chain and $d(B)$ is the position of the block that belongs to the unwritten area. 
\end{proof}}
Whenever a block in the unwritten area has an invalid chain (caused by
shrinking and re-increasing $\mathcal D$; determined by \textsc{minbound}), we return 
the initial 
value for
all words of the block.
We want to remark that the structure $M$ can maintain only $O(w)$ versions since it uses a 
dictionary from P{\v a}tra{\c s}cu and 
Thorup~\cite{PatT14}. 
The dictionary is dynamic and supports modification and access in constant time, but only for $O(w)$ 
entries.

We have to make sure that we have only
$O(|U|)$ chains (counting both valid and invalid chains) so that we are able to store them. 
The problem is that a
shrink can invalidate many chains. 
Therefore, we have to \textit{clean-up} our data structure
$\mathcal D$ as follows:
Whenever we add a chain, we check 
the validness of three old chains, i.e., chains with an old version number that are stored at the beginning of 
the unwritten 
area  in $L$. 
Invalid chains are removed whereas
valid chains are assigned to the current version.
To have the time for the clean-up, we also modify the shrink operation such
that, 
after each shrink operation, we make sure that we store at most $|U| / 2$
chains. 
If this is not the case, we run a full initialization. 
By assuming that every insert into $\mathcal D$ pays a
coin, 
this can be done in amortized constant time since the number of chains that
we have is bounded by the number of insert operations.

\subsection{The General Case}\label{sec:general-case}
The goal in this section is to limit the number of versions to $O(w)$
such that the data structure $M$ from Lemma~\ref{lem:dictionary}
can be always
used.  We achieve this goal by \textit{purifying} the chain
information, i.e., for some old version $v'$, we iterate through the chains
with a version larger than 
$v'$, remove invalid data chains and store the
remaining chains under $v'$. Finally,  we take $v'$ as the current version.
We so
remove all versions larger than $v'$, \mbox{which now can be reused.}

Since we have no data structure to find invalid chains out of
a large amount of all chains, we partition the chains of each version into small subgroups 
and use also here the dictionary from P{\v a}tra{\c s}cu
and Thorup~\cite{PatT14} to index the block positions $B$ and $d(B)$ that
represent the chain.  The dictionary from the subgroup with a version $v$ allows us to find a block 
that lies
behind a boundary $n^*_v$ in constant time as long as such a block exists.  

In detail, a subgroup always has space for $\Theta(w)$ chains (block
positions), the version $v$, and a data structure from P{\v a}tra{\c s}cu and
Thorup~\cite{PatT14} (Figure~\ref{fig:subgroup}).  To organize the subgroups
we use a dynamic family $S$ from Corollary~\ref{cor:multiSet}.  For each
version $v$, we create a dynamic set in $S$.  Each set consists of all subgroups of the same version.  
Instead of having a verification chain with a list $L$,
each chained block in the written area has now
a pointer to its subgroup. If the block is indexed in the subgroup, then its chain is verified.

As 
in the previous section we
use the data structure $M$ to store and check
the boundaries.  Additionally, we store the version $v$ and a table $T$. The
table is
used to store, for each version, the number of chains that are currently used.  
The updates of $T$
are not described below explicitly.  Based on the information in
$T$ we determine how many versions we can purify.
The details of the \textsc{purify} operation are described on the next page.

Note that $M$ and $T$ are of $O(w)$ words and the size of $S$ is linear
in the size of the number of chains and thus bounded by $O(|U|)$. Choosing
the block size $b$ large enough, but still $b=O(1)$, we can guarantee that
all data structures fit into $|U|$ unless $|U|=O(w)$.
By running the clean-up of the last section, $M$ and $T$ can be removed 
if to many writing operations shrink the unwritten area of $\mathcal D$ so much that $|U|=O(w)$.

\begin{figure}[h]
	\centering
	\begin{tikzpicture}[
	start chain = going right,
	node distance = 0pt,
	fbox/.style={draw, minimum width=2em, minimum height=2em,outer sep=0pt},
	box/.style={draw, minimum width=3.5em, minimum height=2em,outer sep=0pt, on chain},
	gray/.style={fill=gray!45},
	A/.style={pattern=north west lines, pattern color=gray!20},
	B/.style={pattern=horizontal lines, pattern color=gray!20},
	C/.style={pattern=north east lines, pattern color=black!99},
	dots/.style={minimum width=2em, minimum height=2em,outer sep=0pt, on chain}],
	\node [box,A] (0) {$B_1$};
	\node [box,B] (1) {$d(B_1)$};
	\node [box,A] (2) {$B_2$};
	\node [box,B] (3) {$d(B_2)$};
	\node [dots] (d) {$\cdots$};
	\node [box] (v) {$v$};
	\node [box,C,minimum width=8em] (7) {};
	\node [box] (p) {};
	\node [box] (n) {};
	\begin{scope}[-{Stealth[length = 2.5pt]}]
	\end{scope}
	\draw[decorate,decoration={brace, amplitude=5pt, raise=5pt, mirror}]
	(0.south west) to node[black,midway,below= 15pt] {$w$ chains for blocks $B_1, B_2, \ldots$} 
	($(d.south east) + (-0.5mm,0)$);%
	\draw[decorate,decoration={brace, amplitude=5pt, raise=5pt, mirror}]
	($(v.south west) + (0.5mm,0)$) to node[black,midway,below= 15pt] {version} ($(v.south east) + 
	(-0.5mm, 0)$);%
	\draw[decorate,decoration={brace, amplitude=5pt, raise=5pt, mirror}]
	($(7.south west) + (0.5mm,0)$) to node[black,midway,below= 15pt] {P\&T data s.} ($(7.south east) + 
	(-0.5mm,0)$);%
	\draw[decorate,decoration={brace, amplitude=5pt, raise=5pt, mirror}]
	($(p.south west) + (0.5mm, 0)$) to node[black,midway,below= 15pt] {prev\vphantom{t}} ($(p.south 
	east) + (-0.5mm, 0)$);%
	\draw[decorate,decoration={brace, amplitude=5pt, raise=5pt, mirror}]
	($(n.south west) + (0.5mm, 0)$) to node[black,midway,below= 15pt] {next} (n.south east);%
	\path [-,dashed, draw](p.north) -- ++(0.0,0.3) -|  (-1,0.66) 
	[label=above left:{Lorry}];
	\path [-,dashed, draw](n.north) -- ++(0.0,0.3) -|  (+12,0.66);
	\end{tikzpicture}
	\caption{A subgroup that is embedded as an element of the dynamic doubly linked list containing several block positions of a version $v$ that are indexed with a dictionary.}\label{fig:subgroup}
\end{figure}
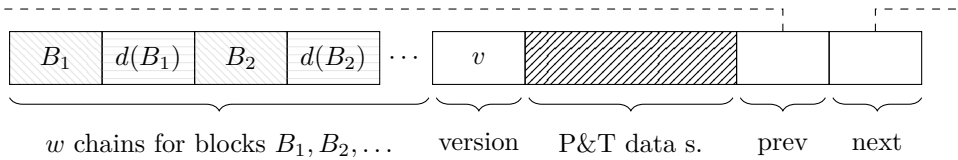
%
%
\begin{itemize}
	\item \textsc{initialize}($n$, $\mathrm{initv}$)
	Allocate $nw + 1$ bits.
		Partition the array into blocks of size $b = 6$. If $n$ is not a multiple of $b$, initialize
                less than $b$
		words and treat them separately.
        Use the last block(s) to store the threshold $t = 0$, initv, and $v=1$ in the unused words.
	Initialize the data structure $M$ and $S$, store their internally required single words
	also in the last block(s) and their sets and lists in parallel (Lemma~\ref{lem:several-data-structures}).
	\item \textsc{read}($i$):
	If $D[0] = 1$, return $D[i]$.
	Else, check to which area $B = \lfloor i/b \rfloor$ belongs.
	If $B$ is inside the written area,
	check if $B$ is verified. 
	If it is, return $\mathrm{initv}$, otherwise $D[i]$.
	
	If $B$ is inside the unwritten area, check if $B$ is chained with a block $d(B)$.
	If it is not, return $\mathrm{initv}$, otherwise proceed with checking if $d(B)$ is verified.
	If it is not, return $\mathrm{initv}$, otherwise follow the verification pointer and read
	the version $v$ out of the subgroup.
	Call $M.\textsc{minbound(v)}$ and return $\mathrm{initv}$ if it returned a boundary that is 
        larger 
	than the block position $B$, 
	otherwise return the right word out of $B$ and $d(B)$.
	
	\item \textsc{write}(i, x):
	If $D[0] = 1$, write at $D[i]$. Otherwise, clean-up three chains as described in the last subsection. 
	Then check to which area $B = \lfloor i/b \rfloor$ belongs.
	If $B$ is inside the written area, check if $B$ is verified.
	If it is not, write at $D[i]$ directly. 
	Otherwise, the block may have a data chain with a block $d(B)$. If so, unchain it (like in 
	\cite{KatG17}) by expanding the written area that gives us a new unused block that
	we use to relocate all values and chains from $B$. Correct also the chains in the subgroups.
	If not, just remove $B$ and $d(B)$ from its subgroup.
	In both cases, $B$ becomes an unused block afterwards. We initialize it and write at $D[i]$ directly.
	 If $B$ is inside the unwritten area, check if $B$ is chained with a block $d(B)$.
	 If it is, check if $d(B)$ is verified and if its chain belonging to a version $v$ is valid 
(test \mbox{M.\textsc{minbound}$(v)\ge B$}).
	 If all is true, then write $x$ at the right position of $B$ and $d(B)$.
	 If the chain is invalid, delete it from its subgroup and move it into the latest subgroup of 
	 the current version.
	 Finally, initialize $B$ and write $x$ as described above.
	 
	 But if $B$ is not chained, expand the written area and chain it with the new block.
	 Write both blocks inside the latest subgroup of the current version. 
	 Set a verification pointer to the subgroup where the chain is stored.
	 As before, initialize $B$ and write \mbox{$x$ as described above.}
	 
	 In all cases, whenever the unwritten area disappears, set  $D[0] = 1$.
	 \item \textsc{increase}($n_\mathrm{old}$, $n_\mathrm{new}$, initv):
	 If $D[0] = 0$, then copy the words of the last block(s) into the new last block(s).
	 Otherwise, initialize the required data structures as described in \textsc{initialize},
	 but set the initial values for the threshold $t$ to $t = \lfloor n_\mathrm{old} / b \rfloor$.
	 If $|U| = O(w)$, use the bit vector solution described in the beginning of 
	 Section~\ref{sec:dyn-arrays}.
	 \item \textsc{purify}($n_\textrm{new}$):	
Iterate from the current version $v = \lceil w \rceil$ in $S$ down and add up the number of chains
stored under a version. Stop at the first version $v'$ that has fewer chains as twice the total number 
of chains of 
all 
versions visited before.
Consider versions $v' + 1$ to $v$ and iterate over all subgroups with that version.
In each subgroup, check for an invalid chain ($B$, $d(B)$) by using $M$,
remove it from the subgroup, from the index of the subgroup, and initialize $B$ with initv and repeat this on the 
subgroup until it has no invalid chains.
In the special case where the subgroup has less than $w$ chains, 
clear the subgroup by moving all chains into 
the latest subgroup of the version $v'$.

Then, change the version number of this subgroup to $v'$,
unlink it from its old set and link it into the set of $v'$.
When finished, set the current version number $v$ to $v'$.
	 \item \textsc{shrink}($n_\mathrm{old}, n_\mathrm{new}$):
	 If the new size of $\mathcal{D}$ either cuts the space used to store 
         the data structures $S$ or $M$ or 
	 leaves less than 
	 $\Theta(w)$ unused words in the unwritten area, run the full initialization.
	 
	 If we do not fully initialize $\mathcal D$, we proceed as follows:
	 If we have $\lceil w \rceil$ versions, call $\textsc{purify}$.
         Otherwise, create a new set inside $S$ and increment the current version 
	 number $v$ by one and
	 check if the last used subgroup has less than $w$ entries. 
         If so, we reuse this subgroup by removing all invalid chains. (In \textsc{purify} we already described the 
         removal.)
\end{itemize}
Who pays for the purification? The idea is to give 
a gold coin to each previous set in $S$, whenever we insert an element.
Whenever a set of size $x$ has at least $x/2$ coins, we clean up the set
and all sets with a larger version number. 
Algorithmically we can not check fast enough if a set already has enough coins. 
Therefore, we wait with the purification until we have $\lceil w \rceil$ versions and check
then the condition for 
every version
from the largest version to the smallest until we found the first version with enough coins.
\begin{theorem}\label{thm:dynArr}
 There is a dynamic array that works in-place and supports
 the operations \textsc{initialize}, \textsc{read},
 \textsc{write} as well as \textsc{increase} in constant time and \textsc{shrink} in amortized
 constant time.
\end{theorem}

\long\def\thmDynArr{
\begin{proof} It remains to show the running times of the operations.
Take $\psi$ as the total number of verification chains similar to Subsection~\ref{sec:complex-solution},
$c_i$ as the number of chains with the version 
$i\in \{0,\ldots,v-1\}$ and $g = \sum_{i=0}^{v - 1} \frac{6i \cdot c_i}{w}$. 
Moreover, choose $\tau$ as the total number of missing elements such that all subgroups 
are full and \mbox{$f=(1-D[0])\max\{0,2w - |U|\}$.}
We use the following potential function $\phi = \psi + v  + g+ \tau + f$.
Note that 
$g$ corresponds to the 
usage of the gold coins 
and $f$ is non-zero only if there are very few blocks in the unwritten area $U$, but $\mathcal
D$ is not completely initialized.

\begin{itemize}
	\item \textsc{initialize}:
	We have to set the threshold $t$ and the current version in $O(1)$ time. Thus,
	$\phi = v = 1$---note that $f=0$ by using the bit-vector solution if necessary.
	\item \textsc{read}:
	Checking chains, 
        reading the version number
	and calling the \textsc{minbound} function of $M$
	can be done in constant time and $\phi$ does not change.
	\item \textsc{write}:
        We possibly have to check if a block is chained, verified and valid.
	We also may create a chain, a verification and insert it into $S$.
	All these operations require $O(1)$ time. 
	In total, the number of chains may increase by one $\Delta \psi = O(1)$, $\Delta g =
	\frac{6v}{w}$ with~$v$ 
	being the current version. Thus, $\Delta\phi=O(1)$.
	\item \textsc{increase}:
	Copying a few blocks of constant size 
        runs in $O(1))$ time and $\Delta \phi=0$ since $\Delta f=0$. 

	\item \textsc{purify}:	
	Whenever we run \textsc{purify} we remove the last $v-i$ versions for the largest $i$ with
	$v_i/2 \le y := \sum_{i+1}^{v - 1}c_i$ is valid. 
	Thus, we first have to search for version $i$, i.e, we consider $v-i$ versions.
	Moreover, we then have to iterate over $3y$ chains in $6y/w$ subgroups---a subgroup is always 
	at least half full.
	In each subgroup we spend $1 + d$ time, where $d$ is the number of deleted chains in
	each subgroup.
	This can be done in at most $(v-i)+6y/w + d^*$ time units where
	$d^*$ is the total number of deleted chains.
	In addition, we may have deleted half empty subgroups and moved their chains to the latest subgroup. This 
	can be done in $O(m)$ time if $m$ is the number of moved chains.

	We can pay for the running time since the value of the potential function decreases as follows.
	By deleting versions larger than $i$, $\Delta v = -(v-i)$.
	Since we change the version of all subgroups with a version larger
	than
        $i$ to $i$, i.e.,
	we shrink the version of $y$ chains by at least one, $\Delta g = -6y/w$.
	In addition, $\Delta(\psi + \tau) \le -2d^*+(d^*-m) = -d^*-m$ where the~$-m$ 
        comes from the fact that half empty subgroups disappear, i.e,
	are not taken into account in $\tau$.
	To sum up, the amortized running time is $O(1)$.

	\item \textsc{shrink}: In Subsection~\ref{sec:complex-solution}, we already analyzed the 
         case that $\mathcal D$ is fully initialized---the only change is that
 we have further parts in the potential function, but their change is either zero or
 negative. 
        Otherwise, we have to delete invalid chains in the latest subgroup
        and update it to the new version, which can be done 
        in $O(1)$ amortized time since $\Delta \psi$ drops linear in the number of
        deleted chains. Moreover, we have to add a new version to $M$
        and possibly delete several older versions. All this can be done in $O(1)$
        amortized time. And finally, we may run a \textsc{purify}; again in $O(1)$     
        amortized time.	
\end{itemize}
We thus have shown all running times as promised in the theorem.
\end{proof}
}

\full{\thmDynArr}

\full{
\section{Acknowledgments}
Andrej Sajenko acknowledges the support
of the Deutsche Forschungsgemeinschaft (DFG), project SpaceeffGA
(KA 4663/1-1). 
}

\phantomsection
\addcontentsline{toc}{chapter}{Bibliography}
\bibliography{main}

\conf{
\clearpage
\phantomsection
\addcontentsline{toc}{chapter}{Appendix}
\appendix

\section*{Appendix}

The purpose of this appendix is to present the missing proofs. The
corresponding
lemma and theorem are repeated for the reader's convenience.

\renewcommand{\thetheorem}{\ref{lem:dictionary}}
\begin{lemma} 
        $M$ can be implemented such that it uses $O(w)$ words, 
        $\textsc{add}$ runs in amortized constant time, and 
        $\textsc{minbound}$ and \textsc{remove} run in constant
        time as long as there are only $O(w)$ versions.
\end{lemma}

\lemDictionary

\renewcommand{\thetheorem}{\ref{thm:dynArr}}
\begin{theorem}
 There is a dynamic array that works in-place and supports
the operations \textsc{initialize}, \textsc{read},
\textsc{write} as well as \textsc{increase} in constant time and \textsc{shrink} in amortized
constant time.
\end{theorem}

\thmDynArr
}

\end{document}